\documentclass[letterpaper, 11pt]{article}
\usepackage{fullpage}

\usepackage{graphicx}
\usepackage{amsmath, amsthm, amssymb, comment, mathtools, url, todonotes}
\usepackage[linesnumbered,ruled,vlined]{algorithm2e}
\usepackage[unicode]{hyperref}
\usepackage{booktabs}

\title{Tight Approximation for~Unconstrained~XOS~Maximization}
\author{
  Yuval Filmus\thanks{Technion - Israel Institute of Technology, Haifa, Israel. Email: {\tt yuvalfi@cs.technion.ac.il}} \and
  Yasushi Kawase\thanks{Tokyo Institute of Technology, Tokyo 152-8550, Japan. Email: {\tt kawase.y.ab@m.titech.ac.jp}} \and
  Yusuke Kobayashi\thanks{Kyoto University, Kyoto 606-8502, Japan. Email: {\tt yusuke@kurims.kyoto-u.ac.jp}} \and
  Yutaro Yamaguchi\thanks{Kyushu University, Fukuoka 819-0395, Japan. Email: \texttt{yutaro\_yamaguchi@inf.kyushu-u.ac.jp}}}
\date{\empty}

\newcommand{\argmax}{\mathop{\rm arg\,max}}

\newcommand{\ot}{\leftarrow}
\newcommand{\hV}{{\hat{V}}}
\newcommand{\hn}{{\hat{n}}}
\newcommand{\hf}{{\hat{f}}}
\newcommand{\tn}{{\tilde{n}}}
\newcommand{\Xalg}{{\tilde{X}}}
\newcommand{\Xopt}{{X^*}}
\newcommand{\Add}{\texttt{Add}\xspace}
\newcommand{\GS}{\texttt{GS}\xspace}
\newcommand{\Submod}{\texttt{SubM}\xspace}
\newcommand{\FracSubadd}{\texttt{FSubA}\xspace}
\newcommand{\Subadd}{\texttt{SubA}\xspace}
\newcommand{\XOS}{\texttt{XOS}\xspace}

\newtheorem{theorem}{Theorem}
\newtheorem{lemma}[theorem]{Lemma}

\newtheorem{claim}[theorem]{Claim}
\newtheorem{corollary}[theorem]{Corollary}
\theoremstyle{definition}
\newtheorem{definition}[theorem]{Definition}
\theoremstyle{remark}
\newtheorem*{remark}{Remark}

\begin{document}
\maketitle
\thispagestyle{empty}

\begin{abstract}
A set function is called XOS if it can be represented by the maximum of additive functions.
When such a representation is fixed, the number of additive functions required to define the XOS function is called the width.

In this paper, we study the problem of maximizing XOS functions in the value oracle model.
The problem is trivial for the XOS functions of width $1$ because they are just additive, but it is already nontrivial even when the width is restricted to $2$.
We show two types of tight bounds on the polynomial-time approximability for this problem.
First, in general, the approximation bound is between $O(n)$ and $\Omega(n / \log n)$, and exactly $\Theta(n / \log n)$ if randomization is allowed, where $n$ is the ground set size.
Second, when the width of the input XOS functions is bounded by a constant $k \geq 2$, the approximation bound is between $k - 1$ and $k - 1 - \epsilon$ for any $\epsilon > 0$.
In particular, we give a linear-time algorithm to find an exact maximizer of a given XOS function of width $2$, while we show that any exact algorithm requires an exponential number of value oracle calls even when the width is restricted to $3$.
\end{abstract}


\clearpage
\thispagestyle{empty}
\tableofcontents
\clearpage
\setcounter{page}{1}

\section{Introduction}
Maximizing a set function is a fundamental task in combinatorial optimization as well as algorithmic game theory.
For example, when an agent has a valuation $v \colon 2^V \to \mathbb{R}$ on an item set $V$, a \emph{demand} of the agent under prices $p\in\mathbb{R}^V$ is a bundle that maximizes her utility, and computing a demand amounts to maximizing a set function $f(X)\coloneqq v(X)-\sum_{x\in X}p_x$.
We remark that, even if the valuation is \emph{monotone} (i.e., $X\subseteq Y \Longrightarrow v(X)\le v(Y)$), the utility $f$ may not be monotone.

One of the most extensively studied classes of set functions in the context of maximization
is the class of \emph{submodular functions}, which naturally captures the so-called \emph{diminishing returns property}.
Buchbinder et al.~\cite{BFNS2015} gave a very simple, randomized 2-approximation algorithm for maximizing nonnegative submodular functions in the \emph{value oracle model}, in which we can access a set function only by querying the oracle for the function value of each set.
This is tight in the sense that an exponential number of value oracle calls are required to achieve $(2-\epsilon)$-approximation (in expectation) for any positive constant $\epsilon$~\cite{FMV2011}, and moreover the 2-approximation algorithm was derandomized later in \cite{BF2018}.
Meanwhile, submodular functions can be exactly minimized in polynomial time in the value oracle model~\cite{GLS1981,GLS1988,IFF2001,Schrijver2000}.

In this paper, we study the maximization problem for another basic class of set functions called \emph{XOS functions}\footnotemark{}, which generalize submodular functions (see also Appendix~\ref{app:class}).
\footnotetext{XOS stands for XOR-of-OR-of-Singletons, where XOR means max and OR means sum. While XOS functions are assumed to be monotone in most existing literature, we allow non-monotone XOS functions. We remark that our problem (unconstrained non-monotone XOS maximization) and monotone XOS maximization under some constraint seem quite different, and it is difficult to compare the results across the two problems.}
A set function $f\colon 2^V\to\mathbb{R}$ is called XOS if it can be represented by the maximum of \emph{additive} functions, i.e., there are set functions $f_i \colon 2^V\to\mathbb{R}$ $(i \in [k] \coloneqq \{1, 2, \dots, k\})$ with $f_i(X) = \sum_{v \in X}f_i(v)$ for each $X \subseteq V$ such that
\begin{align*}
    f(X)=\max_{i\in[k]} f_i(X) = \max_{i\in[k]}\sum_{v\in X}f_i(v) \quad(\forall X\subseteq V),
\end{align*}
where $f_i(v)$ means $f_i(\{v\})$ (we often denote a singleton $\{x\}$ by its element $x$).
We remark that each $f_i$ as well as $f$ can take negative values.
When such a representation is fixed, $k$ is called the \emph{width} of $f$. 
An XOS function admitting a representation of width $k$ is called \emph{$k$-XOS}.
A $1$-XOS function is just an additive function.
The width of an XOS function could be exponential in $|V|$, and we may assume that it is at most $2^{|V|}$.

If we are given an XOS function $f\colon 2^V\to\mathbb{R}$ explicitly as the maximum of additive functions $f_i$ $(i \in [k])$, it is easy to maximize $f$ because we have
\begin{align*}
    \max_{X\subseteq V}f(X)
    =\max_{i\in[k]}\max_{X\subseteq V}f_i(X)
    =\max_{i\in[k]}\sum_{v\in V} \max\{f_i(v),0\}.
\end{align*}
However, if an XOS function $f$ is given by a value oracle like in submodular function maximization, maximization of $f$ becomes nontrivial even when the width of $f$ is restricted to $2$.
Our goal is to clarify what can and cannot be done in polynomial time for maximizing XOS functions given by value oracles.

\subsection*{Our contributions}
The main contribution in this paper is to give two tight bounds on polynomial-time approximation for maximizing XOS functions $f\colon 2^V\to\mathbb{R}$ given by value oracles.
Throughout the paper, we denote by $n$ the cardinality of the ground set $V$.
We also remark that, for simple representation, the running time of each presented algorithm is shown by the number of value oracle calls, which asymptotically dominates the total computational time for the remaining parts.

First, for the general case, we prove that the optimal approximation ratio is almost linear in $n$.
More precisely, we show the following three theorems.
Here, for $\alpha \geq 1$, a set $\tilde{X} \subseteq V$ is called an \emph{$\alpha$-maximizer} if $f(\tilde{X}) \geq \frac{1}{\alpha} \cdot \max_{X\subseteq V}f(X)$.
A deterministic algorithm is said to be an \emph{$\alpha$-approximation algorithm} if it always returns an $\alpha$-maximizer.
A randomized algorithm is so if it is true in expectation, i.e., the expected function value of its output is at least ${1 \over \alpha} \cdot \max_{X \subseteq V}f(X)$.

\begin{theorem}\label{thm:alg_unbounded}
For any $\epsilon > 0$,
there exists a deterministic $(\epsilon n)$-approximation algorithm for XOS maximization that runs in $O(n^{\lceil 1/\epsilon \rceil})$ time.
\end{theorem}

\begin{theorem}\label{thm:alg_unbounded_rand}
For any $\epsilon > 0$,
there exists a randomized $\left({\epsilon n}/{\log n}\right)$-approximation algorithm for XOS maximization that runs in $O(n^{1/\epsilon+1} \log n)$ time.\footnote{In this paper, we always use the natural logarithm without explicitly indicating the base $e$.}
\end{theorem}

On the hardness side, roughly speaking, we show that any $n^{1-\epsilon}$-approximation algorithm requires exponentially many oracle calls and any $o({n}/{\log n})$-approximation algorithm requires super-polynomially many oracle calls.
The following theorem is stated for randomized algorithms, which immediately implies the hardness in the deterministic sense (Corollary~\ref{cor:hard_unbounded}).

\begin{theorem}\label{thm:hard_unbounded}
For any $\epsilon > 0$, there exist a constant $c > 0$ and a distribution of instances of XOS maximization such that any randomized $n^{1 - \epsilon}$-approximation algorithm calls the value oracle at least $\delta \cdot 2^{\Omega(n^c)}$ times with probability at least $1 - \delta$ for any $0 < \delta < 1$.
Moreover, there exists a distribution of instances of XOS maximization such that any randomized $o({n}/{\log n})$-approximation algorithm calls the value oracle at least $\delta \cdot n^{\omega(1)}$ times with probability at least $1 - \delta$ for any $0 < \delta < 1$.
\end{theorem}

As general XOS functions are too hard to maximize, we analyze the problem by restricting the width of the input XOS functions.
When the width is bounded by $k \geq 2$, we prove that the optimal approximation ratio is $k - 1$.
More precisely, we show the following two theorems.

\begin{theorem}\label{thm:alg_bounded}
There exists a deterministic algorithm to find a $(k-1)$-maximizer of a given $k$-XOS function in $O(k^2n)$ time for any $k \ge 2$ (even if $k$ is unknown).
In particular, when $k = 2$, it finds an exact maximizer in $O(n)$ time.
\end{theorem}

\begin{theorem}\label{thm:hard_bounded}
For any $k \geq 3$ and any $\epsilon > 0$, there exist a constant $c > 0$ and a distribution of instances of $k$-XOS maximization such that any randomized $(k - 1 - \epsilon)$-approximation algorithm calls the value oracle at least $\delta \cdot 2^{\Omega(n^c)}$ times with probability at least $1 - \delta$ for any $0 < \delta < 1$.
\end{theorem}

In addition to the $2$-XOS functions,
we also show another special class of $k$-XOS functions that can be exactly maximized in polynomial time (Theorem~\ref{thm:alg_const}).

Finally, we see that the number of \emph{order-different}\footnote{Two set functions $f, g \colon 2^V \to \mathbb{R}$ are \emph{order-equivalent} if [\,$f(X) \leq f(Y)$ if and only if $g(X) \leq g(Y)$\,] for every $X, Y \subseteq V$, and \emph{order-different} otherwise.} XOS functions with bounded width is only single exponential in $n$, specifically $2^{\Theta(n^2)}$ (Theorems~\ref{thm:number_upper} and~\ref{thm:number_lower}), whereas there are doubly exponentially many, specifically $2^{2^{\Theta(n)}}$, order-different XOS functions in general (Corollary~\ref{cor:general_number}). In this sense, $k$-XOS functions look much more tractable than general XOS functions. 
Thus, Theorem~\ref{thm:hard_bounded} is somewhat counterintuitive, because it shows that there is no polynomial-time algorithm for finding a maximizer of a $k$-XOS function even when $k = 3$ if it is given by the value oracle.

\subsection*{Related work}
The problem of maximizing monotone set functions under cardinality constraint has received much attention.
For the submodular case, which is a special case of the XOS case, the greedy algorithm is the best possible and returns an $e/(e-1)$-maximizer~\cite{nemhauser1978best,nemhauser1978analysis}.
For the XOS case, no polynomial-time algorithm can achieve $n^{1/2-\epsilon}$-approximation for any fixed $\epsilon>0$ in the value oracle model~\cite{MSV2008,singer2010}.
For the \emph{subadditive} case, which includes the XOS case (see also Appendix~\ref{app:class}),
Badanidiyuru et al.~\cite{BDO2012} gave a tight $2$-approximation algorithm in the \emph{demand oracle model},
in which we can access a set function by asking the function for a demand $S\in\argmax_{X\subseteq V}\{f(X)-\sum_{v\in X}p_v\}$ under each price vector $p\in\mathbb{R}^V$.
It should be noted that their algorithm does not imply a $2$-approximation algorithm in the value oracle model. 

Polynomial-time approximation for submodular function maximization in the value oracle model has been studied under various combinatorial constraints.
Sviridenko~\cite{Sviridenko04} gave an $e/(e-1)$-approximation algorithm for maximizing a monotone submodular function subject to a knapsack constraint.
Kulik et al.~\cite{KST2013} extended the result to the multiple knapsack constraints case.
For maximizing a monotone submodular function subject to a matroid constraint, Calinescu et al.~\cite{CCPV2011} provided an $e/(e-1)$-approximation algorithm.
The last one is randomized and achieves the optimal approximation ratio in expectation, and the current best guarantee by a deterministic algorithm is slightly better than $2$ \cite{BFG2019}.

\subsection*{Organization}
The rest of this paper is organized as follows.
First, in Section~\ref{sec:algorithms}, we prove Theorems~\ref{thm:alg_unbounded}, \ref{thm:alg_unbounded_rand}, and \ref{thm:alg_bounded} by presenting and analyzing polynomial-time approximation (or sometimes exact) algorithms.
Next, in Section~\ref{sec:hardness}, we show the hardness results,
Theorems~\ref{thm:hard_unbounded} and \ref{thm:hard_bounded}.
Finally, in Section~\ref{sec:number}, we discuss the number of order-different XOS functions.
\section{Algorithms}\label{sec:algorithms}
Let $V$ be a nonempty finite set of size $n$.
Throughout this section, for the input XOS function $f \colon 2^V \to \mathbb{R}$, we may assume that $f(v) > 0$ for all $v \in V$, because any $v \in V$ with $f(v) \le 0$ does not contribute to increasing the function values.
This assumption can be tested in linear time, and when it is violated, one can modify the instance just by removing such unnecessary elements.

For each $X \subseteq V$, we define $I(X) \coloneqq \{\, i \mid f_i(X) = f(X) \,\}$.
In addition, for each index $i$, we call $V_i^* \coloneqq \{\, v \in V \mid f_i(v) = f(v) \,\}$ a \emph{clique with respect to $i$}.
A subset of $V$ is called a \emph{clique} if it is a clique with respect to some $i$.

\subsection{Deterministic \texorpdfstring{$(\epsilon n)$}{(epsilon n)}-approximation for general XOS maximization}\label{sec:(en)-approx}
In this section, we prove Theorem~\ref{thm:alg_unbounded}.
In short, for any $\epsilon > 0$, an $(\epsilon n)$-maximizer is obtained just by taking the best one among all subsets of $V$ of size at most $\lceil 1/\epsilon \rceil$.
A formal description is given in Algorithm~\ref{alg:epsilon}, which clearly calls the value oracle $\sum_{i=0}^{\min\{\lceil 1/\epsilon \rceil,\,n\}}\binom{n}{i} = O(n^{\lceil 1/\epsilon \rceil})$ times.
We take the maximum in Line~\ref{line:2-} while successively updating $X \in \mathcal{X}$ in constant time per set, i.e., adopt an appropriate generator for $\cal{X}$ (cf.~\cite[\S~7.2.1.3]{knuth2005art}).

\begin{algorithm}[htb]
  \caption{A deterministic $(\epsilon n)$-approximation algorithm for XOS maximization}\label{alg:epsilon}
  \SetKwInOut{Input}{Input}\Input{An XOS function $f$ on $V$ (with $f(v) > 0$ for all $v \in V$)}
  \SetKwInOut{Output}{Output}\Output{An $(\epsilon n)$-maximizer $\Xalg \subseteq V$ of $f$}
  Let $\mathcal{X} \ot \{\, X \subseteq V \mid |X| \leq \lceil 1/\epsilon \rceil \,\}$\tcc*{Construct an appropriate generator.}
  \Return $\Xalg \in \argmax_{X\in \mathcal{X}}f(X)$\;\label{line:2-}
\end{algorithm}

The following claim completes the proof of Theorem~\ref{thm:alg_unbounded}.

\begin{claim}
  The output $\Xalg \subseteq V$ of Algorithm~\ref{alg:epsilon} is an $(\epsilon n)$-maximizer of $f$.
\end{claim}

\begin{proof}
Let $\Xopt \subseteq V$ be a maximizer of $f$.
If $|\Xopt| \leq \lceil 1/\epsilon \rceil$, then $\Xopt \in \mathcal{X}$ and hence $f(\Xalg) = f(\Xopt)$.
Otherwise, fix an index $i \in I(\Xopt)$ and let $\hat{X}$ be the set of top $\lceil 1/\epsilon \rceil$ items in $\Xopt$ with the highest values according to $f_i$.
Then, we have $f_i(\hat{X})/|\hat{X}|\ge f_i(\Xopt)/|\Xopt|$ and $\hat{X}\in\mathcal{X}$ by $|\hat{X}|=\lceil 1/\epsilon \rceil$. 
Since $|\Xopt| \leq n$, we have\par
\vspace{\abovedisplayskip}
\hfill$\displaystyle
f(\Xalg) \ge f(\hat{X}) \ge f_i(\hat{X}) \ge \frac{|\hat{X}|}{|\Xopt|} \cdot f_i(\Xopt) \ge \frac{1}{\epsilon n}\cdot f_i(\Xopt) = \frac{1}{\epsilon n}\cdot f(\Xopt).$\hfill
\end{proof}

\subsection{Randomized \texorpdfstring{$(\epsilon n/\log n)$}{(epsilon n)}-approximation for general XOS maximization}\label{sec:(en/logn)-approx}
In this section, we prove Theorem~\ref{thm:alg_unbounded_rand}.
We provide a polynomial-time randomized algorithm whose approximation ratio is $\rho\coloneqq \epsilon n/\log n$.
Without loss of generality, we assume $\rho \geq 2e/(e-2)$ since otherwise $n$ is bounded by a constant depending on $\epsilon$ and hence we can find a maximizer in constant time.

A formal description is given in Algorithm~\ref{alg:randomized_general},
which clearly calls the value oracle $O(n^{1/\epsilon+1} \log n)$ times.
Intuitively, our algorithm first guesses the cardinality of a maximizer $\Xopt$, and then returns the best one among polynomially many samples of size $m = \lceil2|\Xopt|/\rho\rceil$.
Since a uniformly random subset of $\Xopt$ of size $m$ has the expected function value $\frac{m}{|\Xopt|} \cdot f(\Xopt) \geq \frac{2}{\rho}\cdot f(\Xopt)$, if such a subset is sampled with probability at least $1/2$, then the output is a $\rho$-maximizer in expectation.
The following claim gives an upper bound on the number of sufficient samples, which completes the proof of Theorem~\ref{thm:alg_unbounded_rand}.

\begin{algorithm}[htb]
  \caption{A randomized $(\epsilon n/\log n)$-approximation algorithm for XOS maximization}\label{alg:randomized_general}
  \SetKwInOut{Input}{Input}\Input{An XOS function $f$ on $V$ (with $f(v) > 0$ for all $v \in V$)}
  \SetKwInOut{Output}{Output}\Output{An $(\epsilon n/\log n)$-maximizer $\Xalg \subseteq V$ of $f$ in expectation}
  Let $\mathcal{X}\ot\emptyset$\;
  \For{$m \ot 1,2,\dots, \lceil (2 \log n) / \epsilon \rceil$}{
    \For{$t\ot 1,2,\dots, \lceil n^{1/\epsilon + 1}\rceil$}{
      Sample $X_{m,t}$ uniformly at random among the subsets of $V$ of size $m$\;
      $\mathcal{X} \ot \mathcal{X}\cup\{X_{m,t}\}$\;
    }
  }
  \Return $\Xalg \in \argmax_{X\in \mathcal{X}}f(X)$\; \label{line:rand_6}
\end{algorithm}

\begin{claim}
  For any subset $\Xopt \subseteq V$ with $|\Xopt| \ge 2$,
  if we sample a subset of $V$ of size $m = \lceil 2|\Xopt| / \rho \rceil$ uniformly at random $\lceil n^{1/\epsilon + 1} \rceil$ times independently, then at least one sample is a subset of $\Xopt$ with probability at least $1/2$.
\end{claim}

\begin{proof}
The probability that a uniformly random subset of $V$ of size $m$ is a subset of $\Xopt$ is
\begin{align*}
\frac{\binom{|\Xopt|}{m}}{\binom{n}{m}}
&=\frac{|\Xopt|}{n}\cdot\frac{|\Xopt|-1}{n-1}\cdot\cdots\cdot\frac{|\Xopt|-m+1}{n-m+1}
\ge \left(\frac{|\Xopt|(1-2/\rho)}{n}\right)^m
\ge \left(\frac{2|\Xopt|}{en}\right)^m\\
&\ge \left(\frac{2|\Xopt|}{en}\right)^{2|\Xopt|/\rho+1}
\ge \frac{1}{n}\cdot \left(\frac{2|\Xopt|}{en}\right)^{\frac{2|\Xopt|\log n}{\epsilon n}}
\ge \frac{1}{n}\cdot \left(\frac{1}{e}\right)^{\frac{\log n}{\epsilon}}
= \frac{1}{n^{\frac{1}{\epsilon}+1}},
\end{align*}
where the second inequality follows from $\rho \geq 2e/(e-2)$ and the last holds as $x^x\ge (1/e)^{1/e}$ for all $x\in[0,1]$.
Hence, at least one among $\lceil n^{1/\epsilon + 1} \rceil$ samples is a subset of $\Xopt$ with probability at least\par
\vspace{\abovedisplayskip}
\hfill$\displaystyle
1-\left(1-\frac{1}{n^{1/\epsilon + 1}}\right)^{n^{1/\epsilon + 1}}
\ge 1-\frac{1}{e}
\ge \frac{1}{2}.$\hfill
\end{proof}
\vspace{\abovedisplayskip}

We remark that we can obtain a $\rho$-maximizer with high probability if we take more samples.
Let $X$ be a uniformly random subset of $\Xopt$ of size $m$.
Then, $f(X)$ is at least $f(\Xopt)/\rho$ with probability at least $1/\rho$, because $f(X)$ is a random variable that takes a value in $[0,f(\Xopt)]$ and whose expectation is at least $2f(\Xopt)/\rho$.
Hence, if we take $\lceil 2\rho \log n \rceil = \lceil 2\epsilon n \rceil$ times more samples, then we obtain a $\rho$-maximizer with probability at least 
\begin{align*}
1-\left(1-\frac{1}{2} \cdot \frac{1}{\rho}\right)^{\lceil 2\rho \log n \rceil}
\ge 1-\left(\frac{1}{e}\right)^{\log n} = 1-\frac{1}{n}.
\end{align*}

It remains open whether this algorithm can be derandomized, i.e., whether one can achieve the approximation ratio $O(n / \log n)$ by a deterministic algorithm.

\subsection{Maximizing \texorpdfstring{$2$}{2}-XOS functions exactly}\label{sec:exact}
In this section, as a step toward $(k-1)$-approximation for $k$-XOS maximization, we present a linear-time algorithm for finding an exact maximizer of a given $2$-XOS function $f$.

The algorithm is formally described in Algorithm~\ref{alg:exact}, which is intuitively as follows.
First, in Lines~\ref{line:exact_1}--\ref{line:exact_3}, it computes a clique $V_1^* = \{\, v \in V \mid f_1(v) = f(v) \,\}$ (cf.~Lemma~\ref{lem:clique}).
If $V_1^* = V$, then it is a maximizer of $f$ because for every $X \subseteq V$,
\[f(X) \leq \sum_{v \in X}f(v) = \sum_{v \in X} f_1(v) \leq \sum_{v \in V}f_1(v) = f_1(V) \leq f(V).\]
Otherwise, it successively computes the other clique $V_2^* = \{\, v \in V \mid f_2(v) = f(v) \,\}$ in Lines~\ref{line:exact_5}--\ref{line:exact_7} (cf.~Lemma~\ref{lem:clique}).
Finally, in Line~\ref{line:exact_8}, it creates a candidate $Y_i$ for a maximizer of $f$ from each clique $V_i^*$ by adding all the elements which have additional positive contributions to the function value $f(V_i^*)$ (cf.~Lemma~\ref{lem:exact}).

\begin{algorithm}[htb]
  \caption{An exact algorithm for $2$-XOS maximization}\label{alg:exact}
  \SetKwInOut{Input}{Input}\Input{A $2$-XOS function $f$ on $V$ (with $f(v) > 0$ for all $v \in V$)}
  \SetKwInOut{Output}{Output}\Output{An exact maximizer $\Xalg \subseteq V$ of $f$}
  \SetKwFor{ForEach}{for each}{do}{endfall}
  Pick $v_1\in V$ and let $V_1 \ot\{v_1\}$\;\label{line:exact_1}
  \ForEach{$u\in V - v_1$}{
    \lIf{$f(V_1+u)=f(V_1)+f(u)$}{let $V_1 \ot V_1 + u$}\label{line:exact_3}
  }  
  \lIf{$V_1 = V$}{\Return $V_1$}
  Pick $v_2 \in V \setminus V_1$ and let $V_2 \ot\{v_2\}$\;\label{line:exact_5}
  \ForEach{$u\in V - v_2$}{
    \lIf{$f(V_2+u)=f(V_2)+f(u)$}{let $V_2 \ot V_2 + u$}\label{line:exact_7}
  }
  Let $Y_{i}\ot V_i\cup\{\, v\in V \setminus V_i \mid f(V_i+v)>f(V_i) \,\}$ for each $i = 1, 2$\;\label{line:exact_8}
  \Return $\Xalg \in \argmax_{X \in \{Y_1, Y_2\}}f(X)$\;
\end{algorithm}

The running time is clearly bounded by $O(n)$, and the correctness is assured as follows.
First, we see that $V_1$ and $V_2$ computed in Algorithm~\ref{alg:exact} are indeed the cliques.

\begin{lemma}\label{lem:clique}
At the end of Algorithm~\ref{alg:exact}, $V_1$ is a clique, and $V_2$ is the other clique with $V_1 \cup V_2 = V$ if it is computed (i.e., unless $V = V_1$).
\end{lemma}

\begin{proof}
Each $V_i$ $(i = 1, 2)$ is created as a singleton $\{v_i\}$ in Line~\ref{line:exact_1} or \ref{line:exact_5}, and successively updated by adding $u \in V - v_i$ if $f(V_i + u) = f(V_i) + f(u)$ in Line~\ref{line:exact_3} or \ref{line:exact_7}.
The condition is satisfied if and only if $I(V_i) \cap I(u) \neq \emptyset$, and if satisfied, then $I(V_i + u) = I(V_i) \cap I(u)$ holds.
Hence, after the iteration, $I(V_i) = \bigcap_{u\in V_i} I(u)$ and $I(V_i) \cap I(u') = \emptyset$ for each $u' \in V \setminus V_i$.
This means that $V_i = V_{i'}^*$ for each $i' \in I(V_i)$.
When $V_1 \neq V$, since $v_2 \in V_2$ is picked out of $V_1$, we have $V_2 \neq V_1$.
Moreover, since at least one of $f(v) = f_1(v)$ and $f(v) = f_2(v)$ holds for each $v \in V$ by definition, we have $V_1 \cup V_2 = V$.
\end{proof}

Suppose that $V_1 = V_1^* \neq V$ (where exchange $V_1$ and $V_2$ if necessary),
and then the following lemma implies the correctness of Algorithm~\ref{alg:exact}.

\begin{lemma}\label{lem:exact}
At the end of Algorithm~\ref{alg:exact}, $Y_1$ or $Y_2$ is a maximizer of $f$.
\end{lemma}

\begin{proof}
Let $\Xopt \subseteq V$ be a maximizer of $f$.
Without loss of generality, we may assume that $\Xopt = \{\, v \in V \mid f_1(v) > 0 \,\}$ by symmetry, and we show that then $Y_1 = \Xopt$.
Since $V_1 = V_1^* = \{\, v \in V \mid f_1(v) = f(v) > 0 \,\}$, we have $V_1 \subseteq \Xopt$.
In addition, since $f(V_1 + v) \geq f_1(V_1 + v) = f_1(V_1) + f_1(v) > f_1(V_1) = f(V_1)$ for each $v \in \Xopt \setminus V_1$, we have $\Xopt \subseteq Y_1 = V_1 \cup \{\, v \in V \setminus V_1 \mid f(V_1 + v) > f(V_1) \,\}$.

To show $Y_1 = \Xopt$, suppose to the contrary that there exists $v \in Y_1 \setminus \Xopt$.
Then, $v \in Y_1 \setminus V_1$ implies $f(V_1 + v) > f(V_1) = f_1(V_1)$, and $v \not\in \Xopt$ implies $f_1(v) \leq 0$ and hence $f_1(V_1) \geq f_1(V_1 + v)$.
Thus, we have $f(V_1 + v) > f_1(V_1 + v)$, which implies $f_2(V_1 + v) = f(V_1 + v) > f_1(V_1)$.
Since $\Xopt \setminus V_1 \subseteq V_2 = V_2^*$ implies $f_2(\Xopt \setminus V_1) \geq f_1(\Xopt \setminus V_1)$, we have
\[f(\Xopt + v) \geq f_2(\Xopt + v) = f_2(V_1 + v) + f_2(\Xopt \setminus V_1) > f_1(V_1) + f_1(\Xopt \setminus V_1) = f_1(\Xopt) = f(\Xopt),\]
which contradicts that $\Xopt$ is a maximizer of $f$.
\end{proof}

\subsection{\texorpdfstring{$(k-1)$}{(k-1)}-Approximation for \texorpdfstring{$k$}{k}-XOS maximization}\label{sec:(k-1)-approx}
In this section, we prove Theorem~\ref{thm:alg_bounded}.
That is, we present a deterministic $(k-1)$-approximation algorithm for maximizing $k$-XOS functions $f$ that runs in $O(k^2n)$ time for any $k \ge 2$ (even if $k$ is unknown).
In particular, when $k = 2$, it almost coincides with Algorithm~\ref{alg:exact}.
In addition, when $k = o(n)$, it achieves a better approximation ratio than Algorithm~\ref{alg:epsilon} in subcubic time.

The algorithm is shown in Algorithm~\ref{alg:approx_simplified}.
It is worth remarking that the algorithm does not use the information of the width $k$.
As with Algorithm~\ref{alg:exact}, it first computes a family of cliques $\mathcal{V}$ that covers $V$. 
As we will see below, $\mathcal{V}$ contains a $k$-maximizer (more precisely, it contains a $|\mathcal{V}|$-maximizer). 
The difficulty is to improve $k$-approximation to $(k-1)$-approximation.
To resolve this, the algorithm enumerates polynomially many candidates $\mathcal{Z}$ that would be good in addition to $\mathcal{Y}$ like in Algorithm~\ref{alg:exact}.

\begin{algorithm}[htb]
  \caption{A $(k-1)$-approximation algorithm for $k$-XOS maximization}\label{alg:approx_simplified}
  \SetKwInOut{Input}{Input}\Input{A $k$-XOS function $f$ on $V$ (with $f(v) > 0$ for all $v \in V$)}
  \SetKwInOut{Output}{Output}\Output{A $(k-1)$-maximizer $\Xalg \subseteq V$ of $f$}
  Let $R\ot \emptyset$ and $\ell\ot 0$\;
  \While{$R\ne V$}{\label{line:2}
    Let $\ell\ot \ell+1$\;
    Pick $v\in V \setminus R$ and let $V_\ell\ot\{v\}$ and $R\ot R + v$\;\label{line:4}
    \ForEach{$u\in V - v$}{
      \lIf{$f(V_\ell+u)=f(V_\ell)+f(u)$}{let $V_\ell\ot V_\ell+u$ and $R\ot R \cup \{u\}$}\label{line:6}
    }    
  }
  For each $i\in[\ell]$, let $Y_{i}\ot V_i\cup\{\, v\in V \setminus V_i \mid f(V_i+v)>f(V_i) \,\}$\;\label{line:7}
  For each $\{i, j\}\in\binom{[\ell]}{2} = \{\, J \subseteq [\ell] \mid |J| = 2 \,\}$ and $v\in V$, let $Z_{ij}^v\ot V_i\cup V_j\cup\{v\}$\;\label{line:8}
  Let $\mathcal{V} \ot \{V_1, \dots, V_\ell\}$, $\mathcal{Y} \ot \{Y_1, \dots, Y_\ell\}$, and $\mathcal{Z} \ot \{\, Z_{ij}^v \mid \{i,j\} \in \binom{[\ell]}{2},~v \in V \,\}$\;\label{line:9}
  \Return $\Xalg \in \argmax_{X\in \mathcal{V} \cup \mathcal{Y} \cup \mathcal{Z}}f(X)$\;\label{line:10}
\end{algorithm}

We first analyze the running time and then show the correctness. In what follows, let $\ell$ denote its value at the end of Algorithm~\ref{alg:approx_simplified}.

\begin{lemma}\label{lem:ell}
At the end of Algorithm~\ref{alg:approx_simplified}, each $V_i$ $(i \in [\ell])$ is a clique.
In particular, $\ell \leq k$ holds.
\end{lemma}

\begin{proof}
The first part is proved in the same way as Lemma~\ref{lem:clique}.
Since $V_i$ $(i \in [\ell])$ are pairwise distinct due to the choice of $v$ in Line~\ref{line:4} and update of $R$ in Lines~\ref{line:4} and \ref{line:6}, we conclude $\ell \le k$.
\end{proof}

The following two lemmas complete the proof of Theorem~\ref{thm:alg_bounded}.

\begin{lemma}\label{lem:time}
Algorithm~\ref{alg:approx_simplified} can be implemented to run in $O(k^2n)$ time.
\end{lemma}

\begin{proof}
For the while-loop (Lines~\ref{line:2}--\ref{line:6}), the number of iterations is $\ell \leq k$ (Lemma~\ref{lem:ell}).
In each iteration step, the algorithm chooses an element $v \in V$ and just checks whether $f(X + u) = f(X) + f(u)$ or not for some $X \subseteq V$ once for each element $u \in V - v$.
It requires $O(n)$ time (including $O(n)$ value oracle calls), and hence $O(kn)$ time in total.

In Line~\ref{line:7}, the algorithm computes $Y_i \setminus V_i = \{\, v \in V \setminus V_i \mid f(V_i + v) > f(V_i) \,\}$ for each $i \in [\ell]$.
It takes $O(n)$ time (including $O(n)$ value oracle calls) for each $i$, and hence $O(kn)$ time in total.

In Line~\ref{line:8} (for $\mathcal{Z}$), instead of keeping all $Z_{ij}^v$ directly, we first construct $V_i \cup V_j$ $(\{i, j\} \in \binom{[\ell]}{2})$ in $O(k^2n)$ time.
Then, each $Z_{ij}^v \neq V_i \cup V_j$ can be successively constructed from $V_i \cup V_j$ in constant time when taking the maximum in Line~\ref{line:10}.

In Lines~\ref{line:9}--\ref{line:10}, the algorithm just finds a maximizer of $f$ over the family $\mathcal{V} \cup \mathcal{Y} \cup \mathcal{Z}$, whose cardinality is at most $\ell + \ell + \binom{\ell}{2}n = O(k^2n)$.

Thus the total computational time is bounded by $O(k^2n)$.
\end{proof}

\begin{lemma}\label{lem:alg_bounded}
Algorithm~\ref{alg:approx_simplified} returns a $(k-1)$-maximizer $\Xalg$ of $f$.
\end{lemma}

\begin{proof}
Let $\Xopt \subseteq V$ be a maximizer of $f$.

If $\ell<k$, then we have
\begin{align*}
f(\Xalg) &\ge \max_{i\in[\ell]}f(V_i)
\ge \frac{1}{\ell}\cdot\sum_{i\in[\ell]}f(V_i)
= \frac{1}{\ell}\cdot\sum_{i\in[\ell]}\sum_{v \in V_i}f(v)
\ge \frac{1}{\ell}\cdot\sum_{v\in V}f(v)
\ge \frac{1}{\ell}\cdot f(\Xopt)
\ge \frac{1}{k-1}\cdot f(\Xopt),
\end{align*}
where note that $f(V_i)=\sum_{v \in V_i}f(v)$ by Lemma~\ref{lem:ell}, $f(v)>0$ for each $v\in V$ by the assumption, and $\bigcup_{i \in [\ell]} V_i = V$ due to the condition of the while-loop (Line~\ref{line:2}).
Hence, $\Xalg$ is indeed a $(k-1)$-maximizer.

By Lemma~\ref{lem:ell}, in what follows, we consider the case when $\ell=k$, and let us relabel the indices of $V_i$ $(i \in [k])$ so that $V_i = V_i^* = \{\, v \in V \mid f_i(v) = f(v) \,\}$ for each $i \in [k]$. 
Without loss of generality, suppose that $\Xopt = \{\, v \in V \mid f_p(v) > 0 \,\}$ for some $p \in [k]$ (i.e., $\max_{X\subseteq V}f(X)=f(X^*)=f_p(X^*)$).
Then, we have $V_{p} \subseteq \Xopt \subseteq Y_p$ (recall the proof of Lemma~\ref{lem:exact}).

\paragraph{Case 1:}
Suppose that $X^* = Y_p$.
We then have $f(\Xalg) \geq f(Y_p) = f(\Xopt) \geq f(\Xalg)$, and hence the output $\Xalg$ is also a maximizer of $f$.

\paragraph{Case 2:}
Suppose that $X^* \subsetneq Y_p$.
Fix any $v\in Y_p\setminus X^*$ and any $q \in I(V_p + v)$.
Since
\[f_p(V_p + v) = f_p(V_p) + f_p(v) \le f_p(V_p) = f(V_p) < f(V_p + v) = f_q(V_p + v),\]
we have $q \ne p$. 
Then, we have
\begin{align*}
f(\Xopt)
=f_{p}(\Xopt)&=f_{p}(V_p)+f_{p}((\Xopt \setminus V_p) \cap V_q)+f_{p}(\Xopt \setminus (V_p\cup V_q))\\
&< f_{q}(V_p+v)+f_{q}((\Xopt \setminus V_p) \cap V_q)+\sum_{i\in [k]\setminus\{p,q\}}f_{i}(V_i)\\
&= f_q(V_p \cup (\Xopt \cap V_q) \cup \{v\}) + \sum_{i\in [k]\setminus\{p,q\}}f_{i}(V_i)\\
&\le f_{q}(Z_{pq}^v)+\sum_{i\in [k]\setminus\{p,q\}}f_{i}(V_i) \le (k-1)\cdot f(\Xalg),
\end{align*}
where the first inequality holds since $f_p(V_p)<f_q(V_p+v)$ and
\[f_{p}(\Xopt \setminus (V_p\cup V_q))\le \sum_{u\in\Xopt \setminus (V_p\cup V_q)}f(u)\le \sum_{i\in[k]\setminus\{p,q\}}f_i(V_i).\]
Thus, $\Xalg$ is indeed a $(k-1)$-maximizer.
\end{proof}

\subsection{Finding all maximal cliques and its application}
In this section, we show another special class of $k$-XOS functions that can be maximized exactly in polynomial time.
In particular, we prove the following theorem.

\begin{theorem}\label{thm:alg_const}
There exists a deterministic algorithm to find an exact maximizer of a given $k$-XOS function $f$ with the condition
\begin{quote}
$(*)$ for every $v \in V$ and every $i \in [k]$, either $f_i(v) = f(v)$ or $f_i(v) \leq 0$
\end{quote}
in $O(n^{k+1})$ time for any $k \geq 2$ (even if $k$ is unknown).
\end{theorem}

Fix $k \geq 2$ and let $f$ be a $k$-XOS function with the condition $(*)$, i.e., for every $v \in V$ and every $i \in [k]$, either $f_i(v) = f(v) > 0$ or $f_i(v) \leq 0$.

\begin{lemma}
\label{lem:alg_const01}
If an XOS function $f$ satisfies the condition $(*)$, then there exists an inclusion-wise maximal clique that maximizes $f$.
\end{lemma}

\begin{proof}
Let $\Xopt \subseteq V$ be an inclusion-wise minimal maximizer of $f$.
Then, for some $i \in [k]$, we have $\Xopt = \{\, v \in V \mid f_i(v) > 0 \,\} = \{\, v \in V \mid f_i(v) = f(v) \,\} = V_i^*$.
Since $\Xopt$ maximizes $f$, such a clique $V_i^*$ must be inclusion-wise maximal.
\end{proof}

By this lemma, it suffices to find all inclusion-wise maximal cliques.
This can be done by enumerating sufficiently large subsets of cliques and greedily expanding them like Algorithms~\ref{alg:exact} and \ref{alg:approx_simplified}.
The algorithm is formally shown in Algorithm~\ref{alg:maximal}, which does not use the information of the width $k$.

\begin{algorithm}[htb]
  \caption{Finding all maximal cliques}\label{alg:maximal}
  \SetKwInOut{Input}{Input}\Input{A $k$-XOS function $f$ on $V$ (with $f(v) > 0$ for all $v \in V$)}
  \SetKwInOut{Output}{Output}\Output{The family of all inclusion-wise maximal cliques $V_i^*$}
  Let $\mathcal{V}\ot\emptyset$\;
  \For{$\ell=1,2,\dots$}{
    Let $\mathcal{X}_\ell \ot \{\, X \subseteq V \mid f(X) = \sum_{v \in X}f(v),~|X| = \ell \,\}$\;
    \ForEach{$X \in \mathcal{X}_\ell$}{\label{line:2'}
      Let $V_X \ot X$\;\label{line:3'}
      \ForEach{$u\in V \setminus X$}{\label{line:4'}
        \lIf{$f(V_X + u)=f(V_X)+f(u)$}{let $V_X \ot V_X + u$}\label{line:5'}
      }
      $\mathcal{V}\ot \mathcal{V}\cup\{V_X\}$\;\label{line:7'}
    }
    \lIf{$|\mathcal{V}|=\ell$}{\Return $\mathcal{V}$}
  }
\end{algorithm}

\begin{lemma}\label{lem:alg_const02}
For any $k \geq 2$ and any $k$-XOS function $f$, Algorithm~\ref{alg:maximal} returns the family $\mathcal{V}$ of all inclusion-wise maximal cliques in $O(n^{k+1})$ time.
\end{lemma}

\begin{proof}
After the first for-loop with $\ell=1$, it is obvious that $|\mathcal{V}| \ge \ell$. Since we increase $\ell$ by one in each iteration, Algorithm~\ref{alg:maximal} terminates in finite steps. 
In what follows, let $\ell$ denote its value when the algorithm terminates, i.e., $\ell = |\mathcal{V}|$.

We first confirm that each $V_X \in \mathcal{V}$ is indeed a maximal clique.
In Line~\ref{line:3'}, we have $V_X \subseteq V_i^*$ for $i \in I(V_X) = \bigcap_{v \in X}I(v)$ because  $f(X) = \sum_{v \in X}f(v)$ holds.
Hence, as with Lemma~\ref{lem:clique}, in Line~\ref{line:5'}, the condition $f(V_X + u) = f(V_X) + f(u)$ holds if and only if $I(V_X) \cap I(u) \neq \emptyset$, and then $I(V_X + u) = I(V_X) \cap I(u)$.
Thus, after the innermost for-loop (Lines~\ref{line:4'}--\ref{line:5'}), $I(V_X) \subseteq I(u)$ for each $u \in V_X$ and $I(V_X) \cap I(u') = \emptyset$ for each $u' \in V \setminus V_X$.
This means that $V_X = V_i^*$ for each $i \in I(V_X)$, i.e., $V_X$ is a clique. Furthermore, $V_X$ is a maximal clique because $V_X \setminus V_j^* \neq \emptyset$ for each $j \in [k] \setminus I(V_X)$ by the definition of $I(V_X)$.

To show that the output contains all the maximal cliques, suppose to the contrary that some maximal clique $V_i^*$ is not contained in the output $\mathcal{V} = \{V_1, \dots, V_\ell\}$.
Since each $V_j \in \mathcal{V}$ is a clique as shown above, we have $V_i^* \setminus V_j \neq \emptyset$.
Fix any $x_{i,j} \in V_i^* \setminus V_j$ for each $j \in [\ell]$, and let $X_i \coloneqq \{\, x_{i,j} \mid j \in [\ell] \,\}$.
Since $X_i \subseteq V_i^*$, we have $f(X_i) = f_i(X_i) = \sum_{v \in X_i}f_i(v) = \sum_{v \in X_i}f(v)$. 
This shows that $X_i \in \mathcal{X}_{\ell'}$ for some $\ell' \in [\ell]$, because $|X_i| \le \ell$.
Then, we have $V_{X_i} \in \mathcal{V}$.
However, $V_{X_i} \neq V_j$ for each $j \in [\ell]$ because $x_{i,j} \in X_i \setminus V_j \subseteq V_{X_i} \setminus V_j$, which is a contradiction. 

Finally, we analyze the computational time.
The algorithm requires $O(n)$ time to check whether $f(X) = \sum_{v \in X} f(v)$ or not for each $X \subseteq V$ with $|X| \leq \ell = |\mathcal{V}| \leq k$ (recall that every $V_X \in \mathcal{V}$ is a clique $V_i^*$ for some $i \in [k]$), and $O(n)$ time (including $O(n)$ value oracle calls) in Lines~\ref{line:3'}--\ref{line:7'} for each $X \in \mathcal{X}_{\ell'}$ $(\ell' \in [\ell])$.
The number of candidates for $X$ is $\sum_{i = 1}^\ell\binom{n}{i} = O(n^k)$, and hence the total computational time is bounded by $O(n^{k+1})$.
\end{proof}

By Lemmas~\ref{lem:alg_const01} and~\ref{lem:alg_const02}, we obtain Theorem~\ref{thm:alg_const}.
\section{Hardness}\label{sec:hardness}
In this section, we prove two hardness results on XOS maximization (Theorems~\ref{thm:hard_unbounded} and \ref{thm:hard_bounded}), which claim that an exponential (or super-polynomial) number of value oracle calls are required to beat the approximation ratios of Algorithms~\ref{alg:epsilon}, \ref{alg:randomized_general}, and \ref{alg:approx_simplified}.
All hardness results are based on a probabilistic argument.
A key tool is the following lemma.

\begin{lemma}\label{lem:exponential}
Let $\hV=[\hn]$ and let $s,t$ be integers such that $1\le t\le s\le \hn$.
Suppose that we pick, uniformly at random, a set $S\subseteq \hV$ such that $|S|=s$.
Let $\hf\colon 2^\hV\to\mathbb{R}$ be the function defined as
\begin{align*}
    \hf(X)=\begin{cases}
    1&(\text{if } X\subseteq S\text{ and }|X|\ge t),\\
    0&(\text{otherwise}).\end{cases}
\end{align*}
Then, for any positive real $\delta~(<1)$,
any algorithm (including a randomized one) to find $X \subseteq \hV$ with $\hf(X)=1$ calls the value oracle at least $\delta \cdot (\hn/s)^t$ times with probability at least $1 - \delta$.
\end{lemma}

\begin{proof}
Suppose to the contrary that there exists an algorithm to find $X \subseteq \hV$ with $\hf(X) = 1$ that calls the value oracle less than $\delta\cdot(\hn/s)^t$ times with probability more than $\delta$.
By Yao's principle, we may assume that it is deterministic, and suppose that it calls the value oracle for $X_1,X_2,\ldots\subseteq \hV$ in this order.
Note that $f(X_i)=1$ if and only if $S\in \mathcal{X}_i\coloneqq\{\, X\mid X_i\subseteq X\subseteq \hV,~|X|=s \,\}$ and $|X_i|\ge t$.
Since $|\mathcal{X}_i|\le \binom{\hn-t}{s-t}$ holds for any $X_i \subseteq \hV$ with $|X_i| \ge t$, the probability that the algorithm finds $X \subseteq \hV$ with $\hf(X) = 1$ before $m$ oracle calls is at most
\begin{align*}
\frac{\left|\bigcup_{i=1}^m\mathcal{X}_i\right|}{\binom{\hn}{s}}
\le \frac{m\cdot \binom{\hn-t}{s-t}}{\binom{\hn}{s}}
=m\cdot \frac{s}{\hn}\cdot \frac{s-1}{\hn-1}\cdot\dots\cdot \frac{s-t+1}{\hn-t+1}
\le m\cdot \left(\frac{s}{\hn}\right)^t,
\end{align*}
which contradicts that the probability is larger than $\delta$ when $m = \delta\cdot(\hn/s)^t$.
\end{proof}

\subsection{Inapproximability within \texorpdfstring{$n^{1-\epsilon}$}{n-epsilon} and \texorpdfstring{$o(n / \log n)$}{o-n-log-n} for general XOS maximization}
In this section, we prove Theorem~\ref{thm:hard_unbounded}, which is restated for the sake of convenience as follows.

\let\tmp\thetheorem
\renewcommand{\thetheorem}{\getrefnumber{thm:hard_unbounded}}
\begin{theorem}
For any $\epsilon > 0$, there exist a constant $c > 0$ and a distribution of instances of XOS maximization such that any randomized $n^{1 - \epsilon}$-approximation algorithm calls the value oracle at least $\delta \cdot 2^{\Omega(n^c)}$ times with probability at least $1 - \delta$ for any $0 < \delta < 1$.
Moreover, there exists a distribution of instances of XOS maximization such that any randomized $o({n}/{\log n})$-approximation algorithm calls the value oracle at least $\delta \cdot n^{\omega(1)}$ times with probability at least $1 - \delta$ for any $0 < \delta < 1$.
\end{theorem}

\let\thetheorem\tmp
\addtocounter{theorem}{-1}

\begin{proof}
For the first part, let $\epsilon' \coloneqq \epsilon / 2$, and $V = [n]$.
We pick, uniformly at random, a set $S\subseteq V$ such that $|S|=n/2$ (where we assume that $n$ is even).
Suppose that $f(X)=\max \{\, f_i(X) \mid i \in [n + 1] \,\}$, where
\begin{align*}
    f_i(v)&=
    \begin{cases}
    n^{\epsilon'} / 2 &(\text{if }v=i),\\
    0&(\text{if }v\ne i),
    \end{cases}~\text{for}~i\in [n],~~\text{and}~~
    f_{n+1}(v)=
    \begin{cases}
    1&(\text{if }v\in S),\\
    -n&(\text{if }v\not\in S).
    \end{cases}
\end{align*}
We then have $\max_{X\subseteq V}f(X)=f_{n+1}(S)=n/2$ and
\begin{align*}
  f(X)=
  \begin{cases}
    |X| &(\text{if } X\subseteq S\text{ and }|X| > n^{\epsilon'}/2),\\
    0&(\text{if } X=\emptyset),\\
    n^{\epsilon'} / 2&(\text{otherwise}).
  \end{cases}
\end{align*}
Hence, by Lemma~\ref{lem:exponential} (with $\hn = n$, $s = n/2$, and $t = n^{\epsilon'} / 2$), any algorithm to obtain an $n^{1-\epsilon}$-maximizer of $f$ calls the value oracle at least $\delta \cdot 2^{n^{\epsilon'} / 2}$ times with probability at least $1 - \delta$.

For the second part, if we replace $n^{\epsilon'} / 2$ in the definition of $f_i$ with $\tau = \omega(\log n)$, then we derive from Lemma~\ref{lem:exponential} (with $\hn = n$, $s = n/2$, and $t = \tau$) that any algorithm to obtain an $o({n}/{\log n})$-maximizer of $f$ calls the value oracle at least $\delta \cdot n^{\omega(1)}$ times with probability at least $1 - \delta$.
\end{proof}

This theorem shows that an exponential or super-polynomial number of oracle calls are required with high probability, which implies the following corollary. 

\begin{corollary}\label{cor:hard_unbounded}
Let $\epsilon>0$. 
Then, for any randomized $n^{1 - \epsilon}$-approximation (resp.~$o({n}/{\log n})$-approximation) algorithm for XOS maximization, the expected number of value oracle calls is exponential (resp.~super-polynomial) in the ground set size $n$ for the worst instance.
\end{corollary}

\begin{remark}
The hardness result holds also for the problem of maximizing a function that is represented by the maximum of an additive function and a constant, i.e., $f(X)=\max\{g(X), a\}$ $(\forall X \subseteq V)$ for an additive function $g\colon 2^V\to\mathbb{R}$ and a constant $a \in \mathbb{R}$.\footnote{In contrast, for such a function, one can find an $(\epsilon n)$-maximizer and an $(\epsilon n / \log n)$-maximizer in expectation by the same algorithms as Algorithms~\ref{alg:epsilon} and \ref{alg:randomized_general}, respectively.}
To see this, let $\epsilon'$ and $S$ be as in the proof of Theorem~\ref{thm:hard_unbounded} and consider the function $f$ defined as $f(X)=\max\{g(X),n^{\epsilon'}/2\}$, where $g$ is an additive function such that $g(v)=1$ if $v\in S$ and $g(v)=-n$ if $v\not\in S$.
Then, by the same argument, we see that an exponential number of oracle calls are required with high probability to find an $n^{1-\epsilon}$-maximizer. Similarly, a super-polynomial number of oracle calls are required with high probability to find an $o({n}/{\log n})$-maximizer. 
\end{remark}

\subsection{Inapproximability within \texorpdfstring{$k-1-\epsilon$}{k-1-epsilon} for \texorpdfstring{$k$}{k}-XOS maximization}
In this section, we prove Theorem~\ref{thm:hard_bounded}, which we restate here. 

\renewcommand{\thetheorem}{\getrefnumber{thm:hard_bounded}}
\begin{theorem}
For any $k \geq 3$ and any $\epsilon > 0$, there exist a constant $c > 0$ and a distribution of instances of $k$-XOS maximization such that any randomized $(k - 1 - \epsilon)$-approximation algorithm calls the value oracle at least $\delta \cdot 2^{\Omega(n^c)}$ times with probability at least $1 - \delta$ for any $0 < \delta < 1$.
\end{theorem}

\let\thetheorem\tmp
\addtocounter{theorem}{-1}

\begin{proof}
Let $\tn$ be a sufficiently large integer and $\gamma$ be a sufficiently small positive rational number, such that $(k-1) (1-\gamma)^2 > k-1-\epsilon$ holds and $\gamma\tn$ is an integer.
Suppose that $V$ is the union of $k-1$ disjoint sets $V_1, V_2, \dots , V_{k-1}$ with $|V_i|=\tn^{i}$ for each $i \in [k-1]$.
Then, $n = \Theta(\tn^{k-1})$ as well as $\tn = \Theta(n^{1/(k-1)})$. 	
For each $i \in [k-1]$, we pick, uniformly at random, a set $S_i \subseteq V_i$ such that $|S_i|=(1-\gamma)|V_i| = (1 - \gamma)\tn^i$.

Suppose that $f(X)=\max_{i \in [k]} f_i(X)$, where 
\begin{align*}
    f_i(v)&=
    \begin{cases}
    \tn^{k-i} &(\text{if }v\in V_i),\\
    0&(\text{otherwise}),
    \end{cases}\quad \text{for } i\in [k-1],
    \text{ and}\\[2mm]
    f_k(v)&=
    \begin{cases}
    (1-\gamma)\tn^{k-i}&(\text{if }v\in S_i \text{ for some }i \in [k-1]),\\
    -\tn^{k+1}&(\text{otherwise}).
    \end{cases}
\end{align*}
Then, we have $\max_{X\subseteq V}f(X)=f_k(\bigcup_{i=1}^{k-1} S_i) = (k-1)  (1-\gamma)^2\tn^{k} > (k-1-\epsilon) \tn^{k}$.

\begin{claim}
If a nonempty subset $X \subseteq V$ satisfies that $f(X) = f_k(X)$, then there exists $i \in \{2, 3,\dots , k-1\}$ such that $X \cap V_i \subseteq S_i$ and $|X \cap V_i| \ge \frac{\gamma}{k-2} \cdot \tn$. 
\end{claim}

\begin{proof}
Assume that $f(X) = f_k(X)$. Then, it is clear that $X \cap V_i \subseteq S_i$ for each $i \in [k - 1]$. 
Let $j \in [k-1]$ be the minimum index such that $X \cap V_j \neq \emptyset$. 
Since $f_j(X) = \tn^{k-j} |X \cap V_j|$ and $f_k(X) = \sum_{i \ge j} (1-\gamma) \tn^{k-i} |X \cap V_i|$, we derive from $f_k(X) = f(X) \geq f_j(X)$ that
\[\sum_{i > j} (1-\gamma) \tn^{k-i} |X \cap V_i| \ge \gamma \tn^{k-j} |X \cap V_j|.\] 
Since $|X \cap V_j| / (1-\gamma) \ge 1$, this shows that $\sum_{i > j} |X \cap V_i| \ge \gamma \tn$, which implies that $|X \cap V_i| \ge \frac{\gamma}{k-2} \cdot \tn$ for some $i \in \{j+1, j+2, \dots , k-1\}$.  
\end{proof}

This claim shows that we cannot obtain a nonempty subset $X\subseteq V$ with $f(X) = f_k(X)$ unless we find a subset of $S_i$ of size $\frac{\gamma}{k-2} \cdot \tn$ for some $i$.
For fixed $i\in \{2,3,\dots , k-1\}$, by Lemma~\ref{lem:exponential} (with $\hV = V_i$ as well as $\hn = \tn^i$, $S = S_i$ as well as $s = (1 - \gamma) \tn^i$, and $t = \frac{\gamma}{k-2} \cdot \tn$),
any algorithm to find a set $X \subseteq V$ such that $X\cap V_i \subseteq S_i$ and $|X \cap V_i| \ge \frac{\gamma}{k-2} \cdot \tn$ calls the value oracle at least
\[\delta'\cdot\left(\frac{1}{1 - \gamma}\right)^{\gamma\tn/(k-2)} = \delta' \cdot 2^{\Theta(\tn)} = \delta' \cdot 2^{\Theta(n^{1/(k-1)})}\]
times with probability at least $1 - \delta'$, where $0<\delta'<1$.
By setting $\delta = (k-2) \delta'$, we see that any algorithm to find a set $X \subseteq V$ such that $X\cap V_i \subseteq S_i$ and $|X \cap V_i| \ge \frac{\gamma}{k-2} \cdot \tn$ for some $i\in \{2,3,\dots , k-1\}$ calls the value oracle at least $\delta' \cdot 2^{\Theta(n^{1/(k-1)})} = \delta \cdot 2^{\Theta(n^{1/(k-1)})}$ times with probability at least $1 - (k-2) \delta' = 1-\delta$ by the union bound.
The same number of oracle calls are required for obtaining  $X\subseteq V$ such that $f(X)>\tn^{k}$, because $\max_{X \subseteq V} f_i (X) = \tn^k$ for $i\in [k-1]$.
By combining this with $\max_{X\subseteq V}f(X) > (k-1-\epsilon) \tn^{k}$, we complete the proof of Theorem~\ref{thm:hard_bounded}. 
\end{proof}

This theorem shows that an exponential number of oracle calls are required with high probability, which implies the following corollary. 

\begin{corollary}\label{cor:hard_bounded}
Let $k \geq 3$ and $\epsilon>0$.
Then, for any randomized $(k - 1 - \epsilon)$-approximation algorithm for $k$-XOS maximization, the expected number of value oracle calls is exponential in the ground set size $n$ for the worst instance.
\end{corollary}

It is worth mentioning that the hardness results for $k$-XOS maximization (Theorem~\ref{thm:hard_bounded} and Corollary~\ref{cor:hard_bounded}) hold only when $k$ is a fixed constant.
The approximability when $k = \omega(1)$ remains open.
\section{The Number of XOS Functions}\label{sec:number}
In this section, we show that the number of order-different XOS functions with bounded width is single exponential in the ground set size $n$, whereas there are doubly-exponentially many order-different XOS functions in general.
Recall that, for set functions $f,g\colon 2^V\to\mathbb{R}$ on the common ground set $V$, we say that $f$ and $g$ are \emph{order-equivalent} if $f(X)\le f(Y)$ if and only if $g(X)\le g(Y)$ for all $X,Y\subseteq V$, and that $f$ and $g$ are \emph{order-different} if they are not order-equivalent.

First, we observe that the number of order-different XOS functions is doubly exponential in $n$.
In particular, even if we are restricted to the rank functions of matroids, which are normalized and submodular, and hence XOS (see Appendix~\ref{app:class}), the number is so large. For the basics on matroids, we refer the readers to \cite{Oxley}.

\begin{theorem}[Knuth~\cite{Knuth1974}]\label{thm:Knuth}
The number of distinct matroids on $V = [n]$ is $2^{2^{\Theta(n)}}$.
\end{theorem}

\begin{corollary}\label{cor:general_number}
The number of order-different XOS functions on $V=[n]$ is $2^{2^{\Theta(n)}}$.
\end{corollary}

\begin{proof}
As the number of binary relations on $m$ elements is $2^{m^2}$ and the order-equivalence classes of set functions on $V$ correspond one-to-one to the total preorders on $2^V$, the number of order-different XOS functions on $V=[n]$ is at most $2^{(2^n)^2}=2^{2^{O(n)}}$.

The matroids on $V$ are uniquely defined by their rank functions $f \colon 2^V \to \mathbb{R}$, which are XOS (as normalized, monotone, and submodular).
Moreover, if two matroid rank functions $f_1$ and $f_2$ on $V$ are distinct, then there exist $X \subsetneq V$ and $e \in V \setminus X$ such that $f_1(X) = f_2(X) = f_2(X + e) = f_1(X + e) - 1$ (or the symmetric condition obtained by exchanging the indices $1$ and $2$), which implies that $f_1$ and $f_2$ are order-different.
\end{proof}

Next, we show that, for any fixed $k$, the number of order-different $k$-XOS functions is single exponential in $n$.

\begin{theorem}\label{thm:number_upper}
The number of order-different $k$-XOS functions on $V=[n]$ is $2^{O(k^2n^2)}$.
\end{theorem}

\begin{proof}
Let $f$ be a $k$-XOS function with additive functions $f_1,\dots,f_k$ such that $f(X)=\max_{i\in[k]}f_i(X)$ $(\forall X\subseteq V)$.
Fix a function $\iota\colon 2^{V}\to [k]$ that represents a maximizer index, i.e., $\iota(X)\in I(X)$ for each $X\subseteq V$.
Let us consider the following polyhedron $P[f]$:
\begin{align*}
  P[f]\coloneqq\left\{\,
  w ~\middle|~
  \begin{array}{ll}
  \sum_{v\in X}w_{\iota(X),v}-\sum_{u\in Y}w_{\iota(Y),u}\ge 1&
  (\forall X,Y\in 2^{V}\text{ with }f(X)>f(Y)),\\
  \sum_{v\in X}w_{\iota(X),v}-\sum_{u\in Y}w_{\iota(Y),u}= 0&
  (\forall X,Y\in 2^{V}\text{ with }f(X)=f(Y)),\\
  \sum_{v\in X}w_{\iota(X),v}-\sum_{v\in X}w_{i,v}\ge 0&
  (\forall X\in 2^{V},\ \forall i\in[k]\setminus\{\iota(X)\})
  \end{array}
  \right\}.
\end{align*}
Here, we have $kn$ variables and $O(2^{2n})$ linear constraints.
For any feasible weight $w\in P[f]$, the function $g$ defined by $g(X)=\max_{i\in[k]}\sum_{v\in X}w_{i,v}$ for all $X\subseteq V$ is order-equivalent to $f$.

The constraint matrix of the polyhedron $P[f]$ is full-rank, because we have one of $w_{\iota(v),v} \ge 1$, $-w_{\iota(v),v} \ge 1$, and $w_{\iota(v),v}=0$ for any $v\in V$ by the first and second inequalities, and $w_{\iota(v),v} - w_{i,v} \ge 0$ for any $v \in V$ and any $i\in[k]\setminus\{\iota(v)\}$ by the third inequality. 
Also, a vector $w$ defined by $w_{i,v} \coloneqq \beta\cdot f_i(v)$ $(i \in [k],~v \in V)$ with a sufficiently large $\beta$ is in $P[f]$, and hence $P[f]$ is feasible (nonempty).
Thus, $P[f]$ has a basic solution (vertex) (see, e.g., \cite[\S~8.5]{schrijver1998}).

Let $\hat{w}$ be a basic solution of $P[f]$.
Then, by considering the corresponding inequalities, we have $A\hat{w}=b$ (i.e., $\hat{w}=A^{-1}b$) for a nonsingular matrix $A\in\{-1,0,1\}^{kn\times kn}$ and a vector $b\in\{0,1\}^{kn}$.

Thus, for any $k$-XOS function $f$, there exists an order-equivalent $k$-XOS function $g$ that is defined by a weight $\hat{w}\coloneqq A^{-1}b$ with $A\in\{-1,0,1\}^{kn\times kn}$ and $b\in\{0,1\}^{kn}$.
As the number of possible such weights is at most $3^{(kn)^2}\cdot 2^{kn}=2^{O(k^2n^2)}$,
the proof is complete.
\end{proof}

Comparing Corollary~\ref{cor:general_number} and Theorem~\ref{thm:number_upper}, we get that almost all XOS functions on the ground set of size $n$ have width $\Omega(2^n / n)$.
It may be of interest to explore a better lower bound on the threshold $t$ such that most XOS functions have width at least $t$.

It is worth mentioning that the upper bound on the number of order-different $k$-XOS functions given in Theorem~\ref{thm:number_upper} is tight with respect to $n$.

\begin{theorem}\label{thm:number_lower}
The number of order-different additive (1-XOS) functions on $V=[n]$ is $2^{\Omega(n^2)}$.
\end{theorem}

\begin{proof}
Consider additive functions $f$ such that  $f(v)=2^{v-1}$ if $v=1, 2, \dots,\lfloor n/2\rfloor$ and $f(v)\in\{0, 1, 2, \dots,2^{\lfloor n/2\rfloor} - 1\}$ if $v=\lfloor n/2\rfloor+1,\dots,n$.
There are $(2^{\lfloor n/2\rfloor})^{\lceil n/2\rceil}=2^{\Omega(n^2)}$ possibilities and all the functions are order-different (consider binary expansion of $f(v)$ $(v > \lfloor n/2 \rfloor)$ using $f(v)$ $(v \leq \lfloor n/2 \rfloor)$).
Thus the theorem holds.
\end{proof}

\section*{Acknowledgments}
We thank Tomomi Matsui for useful discussion on the number of order-different set functions.
We are grateful to the anonymous reviewers for giving insightful comments and suggestions.
YF is a Taub Fellow --- supported by the Taub Foundations.
His research was funded by ISF grant 1337/16.
YK, YK, and YY were supported by JST ACT-I Grant Numbers JPMJPR17U7 and JPMJPR17UB, and by JSPS KAKENHI Grant Numbers JP15H05711, JP16H03118, JP16K16005, JP16K16010, and JP18H05291.
Most of this work was done when YY was with Osaka University.

\bibliographystyle{plain}
\bibliography{xos}

\appendix
\section{Classes of Set Functions}\label{app:class}
The following classes are well-studied in combinatorial optimization
and recently also in algorithmic game theory as valuations of agents.
\begin{definition}\label{def:set_function}
A set function $f\colon 2^V\to\mathbb{R}$ is 
\begin{itemize}
\item \emph{normalized}, if $f(\emptyset)=0$.
\item \emph{monotone}, if $f(X)\le f(Y)$ for every $X\subseteq Y\subseteq V$.
\item \emph{additive}, if $f(X)=\sum_{v\in X}f(v)$ for every $X \subseteq V$.
\item \emph{gross-substitute}, if 
for every $p,q\in\mathbb{R}^V$ with $p\le q$ and every $X\in\argmax\{\, f(S)-\sum_{v\in S}p_v\mid S\subseteq V\,\}$,
there exists $Y\in\argmax\{\, f(S)-\sum_{v\in S}q_v\mid S\subseteq V \,\}$ such that $\{\, v\in X\mid p_v=q_v \,\}\subseteq Y$.
\item \emph{submodular}, if $f(X)+f(Y)\ge f(X\cup Y)+f(X\cap Y)$ for every $X,Y\subseteq V$.
\item \emph{fractionally subadditive}: $f(T)\le\sum_{i} \alpha_if(S_i)$ for every $T, S_i \subseteq V$ whenever $\alpha_i\ge 0$ and $\sum_{i \colon v\in S_i}\alpha_i\ge 1$ ($\forall v\in T$).
\item \emph{subadditive} if $f(X\cup Y)\le f(X)+f(Y)$ for every $X, Y \subseteq V$.
\end{itemize}
\end{definition}

Let us denote by \Add, \GS, \Submod, \FracSubadd, \Subadd, and \XOS the sets of (normalized) additive functions, of normalized gross-substitute functions,
of normalized submodular functions, of normalized fractionally subadditive functions,
of normalized subadditive functions, and of (normalized) XOS functions, respectively.
Also, we add $*$ to each of them to assume the monotonicity in addition to each property.
We then have the following relations~\cite{Feige2007,GS1999,LLN2006}:
\begin{align*}
&\Add \subseteq \GS \subseteq \Submod \subseteq \XOS,\\
&\Add^* \subseteq \GS^* \subseteq \Submod^* \subseteq \XOS^* = \FracSubadd^*~(= \FracSubadd) \subseteq \Subadd^*.
\end{align*}

\end{document}